\documentclass[12pt,a4paper]{article}

\usepackage[dvips]{graphicx}
\usepackage[cp1251,ctt]{inputenc}
\usepackage{amsmath,amsthm,amsfonts,amssymb,cite}
\usepackage[T2A]{fontenc}

\newtheorem{assertion}{Proposition\!}

\def\Ba{\text{B}}
\def\Fa{\text{F}}
\def\tr{\operatorname{tr}}
\def\Tr{\operatorname{Tr}}
\def\diag{\operatorname{diag}}
\def\defa{\operatorname{def}}

\begin{document}
\title{
$${}$$\\
{\bf The correlation functions} \\
{\bf of the ${\bf XX}$ Heisenberg magnet and}\\
{\bf random walks of vicious walkers }}
\author{
$${}$$\\
{\bf N.~M.~Bogoliubov$^\dagger$, C.~Malyshev$^\ddagger$}\\[0.5cm]
{\small Steklov Mathematical Institute,
St.-Petersburg Department, RAS}\\
{\small  Fontanka 27, St.-Petersburg, 191023, Russia} \\
[0.5cm]
$^\dagger$ e-mail: {\it bogoliub@pdmi.ras.ru}\\
\,\,$^\ddagger$ e-mail: {\it malyshev@pdmi.ras.ru} }

\date{}

\maketitle

\begin{abstract} \noindent
A relationship of the random walks on one-dimensional periodic lattice and the correlation functions of the $XX$ Heisenberg spin chain is investigated. The operator averages taken over the ferromagnetic state play a role of generating functions of the number of paths made by the so-called ``vicious'' random walkers (the vicious walkers annihilate each other provided they arrive at the same lattice site). It is shown that the two-point correlation function of spins, calculated over eigen-states of the $XX$ magnet, can be interpreted as the generating function of paths made by a single walker in a medium characterized by a non-constant number of vicious neighbors. The answers are obtained
for a number of paths  made by the described  walker from some fixed lattice site to another sufficiently remote one. Asymptotical estimates for the number of paths are provided in the limit, when the number of steps is increased.
\end{abstract}

\leftline{\emph{{\bf Keywords:}} random walks, Heisenberg magnet,
correlation functions}

\thispagestyle{empty}

\newpage

\section{Introduction}
\label{sec1}

The random walks is a classical problem both for combinatorics and
statistical physics. The problem of enumeration of the paths made
by the, so-called, \textit{vicious} walkers on the one-dimensional
lattice has been formulated and investigated in details by
Fisher~\cite{1}. It is supposed that any two vicious walkers,
provided both arrive at the same lattice site, annihilate not only
one another but all other walkers as well. The problem mentioned
still continues to attract considerable attention both of
physicists and mathematicians~\cite{2}--\cite{9}. Closely related
problems arise also in the studies of the self-organized
criticality~\cite{10}, domain walls~\cite{11}, and
polymers~\cite{12}. In paper~\cite{13} a random walks of the
annihilating particles on a ring was considered. In
paper~\cite{14} a random turns walks on a semi-axes with a
possible creation of the particles at the origin was studied.

It has been shown in~\cite{15},~\cite{16} that the correlation
functions, obtained as an averages over the ferromagnetic state of
the $XX$ Heisenberg chain, can be used for enumeration of the
paths of random walks of vicious walkers. In the present paper the
averages of special type are investigated both for the case of
ferromagnetic state and for superposition of the eigen-states of
the $XX$ magnet in zero magnetic field. The averages in question
play a role of the generating functions of number of paths of the
vicious walkers. The calculation of the correlation functions is
carried out by means of the functional integration~\cite{17},~\cite{18}. The answers are obtained for the number of paths of a single pedestrian which is travelling from one chosen site to another sufficiently remote lattice site. The asymptotical estimates are obtained for the number of paths in the limit, when the number of steps (and, correspondingly, the number of random turns) is increasing.

The paper is organized as follows. Section~\ref{sec1} has an introductory character. The Hamiltonian of the model and general calculation of the correlation functions are discussed in Section~\ref{sec2}. Section~\ref{sec3} deals with the specific calculations and the corresponding asymptotic estimates. Discussion in Section~\ref{sec4} concludes the paper.

\section{The model and the correlation functions}
\label{sec2}

The $XX$ magnet we are interested in is a particular limit of a
more general spin model known as the $XY$ Heisenberg chain, with
the Hamiltonian  in the transverse magnetic field $h>0$ given
by~\cite{19},~\cite{20}:
\begin{equation}
\begin{aligned}
&H=H_0+\gamma H_1-hS^z,\\
&H_0\equiv-\sum^{M}_{n,m=1}\Delta^{(+)}_{nm}\sigma^+_{n}\sigma^-_{m},
\\
&H_1\equiv-\frac12\sum_{n,m=1}^M\Delta^{(+)}_{nm}
(\sigma^+_n\sigma^+_{m}+\sigma^-_n\sigma^-_{m}),\qquad
S^z\equiv\frac{1}2\sum_{n=1}^M\sigma^z_n,
\end{aligned}
\label{1}
\end{equation}
where $S^z$ is $z$-component of the total spin operator, and $\gamma$
is the anisotropy parameter. The local spin operators
$\sigma^\pm_n=(\sigma^x_n\pm i\sigma^y_n)/2$ and~$\sigma^z_n$ are
given by the Pauli matrices, which depend on the lattice argument
$n\in\mathcal M\equiv\{1,2,\dots,M\}$, where $M=0\pmod{2}$. The corresponding commutation relations have the form:
$$[\sigma^+_k,\sigma^-_l]\,=\,\delta_{k,l}\sigma^z_l\,,\quad
[\sigma^z_k,\sigma^\pm_l]\,=\,\pm2\delta_{k,l}\sigma^{\pm}_l\,.
$$
The introduced \textit{hopping matrix}~$\Delta^{(s)}$
is defined by the following entries:
\begin{equation}
\Delta^{(s)}_{nm}\equiv\frac12\,
(\delta_{|n-m|,1}+s\delta_{|n-m|,M-1}), \label{2}
\end{equation}
where ${\delta}_{n,l}$ is the Kronecker symbol, and $s$ can take two
values: $s=\pm$. It is assumed that the periodic boundary conditions
$\sigma^{\alpha}_{n+M}=\sigma^{\alpha}_n$ are imposed for any $n\in\mathcal M$. The Hamiltonian $H$~\eqref{1} is reduced to the Hamiltonian of the $XX$ magnet at zero value of the parameter $\gamma$.

The most general definition of the time $t$ and temperature $T\equiv
1/\beta$ dependent correlation functions of the model under consideration looks as follows:
\begin{equation}
G_{j;l}^{ab}(t)\equiv\frac{1}Z
\Tr(\sigma^a_{j}(0)\sigma^b_l(t)e^{-\beta H}),\qquad
Z\equiv\Tr(e^{-\beta H}),
\label{3}
\end{equation}
where $\sigma^b_l(t)\equiv e^{itH}\sigma^b_le^{-itH}$ and $\Tr$
means the averaging with respect to all eigen-states of the
Hamiltonian~$H$. In addition, the normalization involves the
partition function~$Z$. Calculation of the correlators~\eqref{3}
has been carried out in~\cite{21} as averaging over all
eigen-functions of the Hamiltonian of the $XX$ magnet. In
\cite{21} the main attention has been paid  to a relationship
between the correlation functions and the Fredholm determinants in
the thermodynamic limit. In the present paper we shall consider
the $XX$ chain only and denote its Hamiltonian by~$H$.

To calculate the averages~\eqref{3} one can use a representation of the
canonical Fermi variables~$c_j$, $c^{\dagger}_j$, $j\in{\mathcal M}$
through the spin variables \cite{19},~\cite{20}. The corresponding
Jordan--Wigner transformation has the form:
\begin{equation}
\sigma^+_n=\biggl(\,\prod_{j=1}^{n-1}\sigma^z_j\biggr)c_n,\qquad
\sigma^-_n=c_n^{\dagger}\biggl(\,\prod_{j=1}^{n-1}\sigma^z_j\biggr),
\qquad n\in{\mathcal M},
\label{4}
\end{equation}
where $\sigma^z_j=1-2c_j^{\dagger} c_j$. The periodic boundary
conditions for the spin variables lead to the following boundary
conditions for the Fermi variables:
\begin{equation}
c_{M+1}=(-1)^{\mathcal{N}}c_1,\qquad
c^{\dagger}_{M+1}=c^{\dagger}_1(-1)^{\mathcal{N}},
\label{5}
\end{equation}
where $\mathcal{N}=\sum_{n=1}^Mc_n^{\dagger} c_n$ is the operator
of the total number of particles. The Hamiltonian $H$~\eqref{1} takes the following form in the fermion representation~\cite{19},~\cite{20}
\begin{equation}
H=H^+P^++H^-P^-,
\label{6}
\end{equation}
where $P^{+}$ ($P^{-}$) are projectors on the states characterized
by an even/odd number of fermions:
$$
P_++P_-=\mathbb I\,,\quad P_+-P_-=(-1)^{\mathcal N}\,.
$$
The operators $H^{\pm}$~\eqref{6} are formally identical, their superscripts  $s=\pm$ point out an appropriate specification of the boundary conditions~\eqref{5}:
$$
c_{M+1}\,=\,-s\,c_1\,,\quad c^{\dagger}_{M+1}\,=\,- s\,c^{\dagger}_1\,.
$$
To put it differently, the quadratic in the fermion variables operators~$H^\pm$ has the following representation:
\begin{equation}
H^{\pm}=c^{\dagger}\widehat H^{\pm} c-\frac{Mh}2,\qquad
\widehat H^{\pm}=-\hat\Delta^{(\mp)}+h{\hat I},
\label{7}
\end{equation}
where the matrices ${\widehat H}^{\pm}$ are expressed through the hopping matrices~\eqref{2} and $\hat I$ is the unit matrix:
$$
{\widehat H}^\pm\,=\,\left(
\begin{array}{cccccc}   h   & -1/2 &     &   &      & \pm1/2 \\
        -1/2  &  h   & -1/2  &        &         &       \\
              & -1/2 &  h    & -1/2   &         &        \\
              &      &       & \dots  &         &        \\
              &      &       &  -1/2  & h       & -1/2   \\
     \pm 1/2  &      &       &        & -1/2    &  h
\end{array}\right)
$$
(only non-zero entries are displayed). Besides, the short-hand notations~$c^{\dagger}$ and~$c$ are used in~\eqref{7} for the $M$-dimensional row and column with the entries~$c_n^{\dagger}$,~$c_n$, $n\in\mathcal M$.

In particular, the correlator~\eqref{3} at $a=b=z$ takes the following form in the representation~\eqref{4}  \cite{18},~\cite{22},~\cite{23}:
\begin{equation}
G_{j;l}^{zz}(t)=1-\frac{2}Z\Tr(c^{\dagger}_{j}c_{j}e^{-\beta H})-
\frac{2}Z\Tr(c^{\dagger}_lc_le^{-\beta H})+
\frac{4}Z\Tr(c^{\dagger}_{j}c_{j}
e^{itH}c^{\dagger}_lc_le^{-(\beta+it)H}).
\label{8}
\end{equation}
In order to calculate~\eqref{8}, it is convenient to
introduce the generating functional~\cite{18}:
\begin{equation}
\mathcal G\equiv\mathcal G(S,T\mid\lambda,\nu)=
\frac{1}Z\Tr(e^Se^{-\lambda H}e^Te^{-\nu H}),
\label{9}
\end{equation}
where~$\lambda$, $\nu$ are the complex parameters, $\lambda+\nu=\beta$.
The quadratic operators $S\equiv c^{\dagger}{\widehat S}c$ and $T\equiv c^{\dagger}{\widehat T}c$, used in~\eqref{9}, are defined by means of the matrices $\widehat S=\diag\{S_1,S_2,\dots,S_M\}$ and $\widehat T=\diag\{T_1,T_2,\dots,T_M\}$. For instance, the last term in right-hand side of  \eqref{8} is obtained from~\eqref{9} in the following way:
\begin{equation}
\lim_{\substack{S_n,T_n\to0,\\n\in \mathcal M}}\,\lim_{\lambda\to-it}\,
\lim_{\nu\to\beta+it}\frac{\partial}{\partial S_j}\,
\frac{\partial}{\partial T_l}\,\mathcal G(S,T\mid\lambda,\nu).
\label{10}
\end{equation}
The trace in right-hand side of~\eqref{9} can be re-written by
means of~\eqref{6}~\cite{18}:
\begin{equation}
\Tr(e^Se^{-\lambda H}e^Te^{-\nu H})=
\frac12\,(\mathcal G^+_{\Fa}Z^+_{\Fa}+\mathcal G^-_{\Fa}Z^-_{\Fa}
+\mathcal G^+_{\Ba}Z^+_{\Ba}-\mathcal G^-_{\Ba}Z^-_{\Ba}),
\label{11}
\end{equation}
where
\begin{equation}
\begin{aligned}
&\mathcal G^{\pm}_{\Fa}Z^{\pm}_{\Fa}\equiv
\Tr(e^Se^{-\lambda H^{\pm}}e^Te^{-\nu H^{\pm}}),
\\
&\mathcal G^{\pm}_{\Ba}Z^{\pm}_{\Ba}\equiv
\Tr(e^Se^{-\lambda H^{\pm}}e^T(-1)^{\mathcal N}e^{-\nu H^{\pm}}),
\end{aligned}
\label{12}
\end{equation}
and
$$
Z^{\pm}_{\Fa}\,=\,\Tr(e^{-\beta H^{\pm}})\,,\quad
Z^{\pm}_{\Ba}\,=\,\Tr((-1)^{\mathcal N}e^{-\beta H^{\pm}})\,.
$$
Moreover, for the partition function $Z$ we obtain the representation: $$
Z=\frac12\,(Z^+_{\Fa}+Z^-_{\Fa}+Z^+_{\Ba}-Z^-_{\Ba})\,.
$$
In the thermodynamic limit the terms with the subscript $\Ba$ are
mutually compensated, therefore, in order to obtain $\mathcal G$~\eqref{9} it is enough to calculate $\mathcal G^{\pm}_{\Fa}$.

The considered fermion representation is characterized by the
existence of the Fock state~$|0\rangle$ common for both operators
$H^{+}$ and $H^{-}$, and satisfying the relations $c_k|0\rangle=0$, $k\in\mathcal M$. However, the corresponding coherent states over~$|0\rangle$,
\[
\begin{array}{rcl}
&&\mid\!z\big\rangle\,\equiv\,\exp \Big(\sum\limits^M_{k=1} c^\dagger_k
z_k\Big) \mid\!0\big\rangle\,\equiv\,\exp (c^\dagger z)\mid\!0\big\rangle\,,\\ [0.5cm]
&&\big\langle z^*\!\mid\,\equiv\,\big\langle 0\!\mid \exp
       \Big(\sum\limits^M_{k=1} z^*_k c_k \Big)\,\equiv\,
     \big\langle 0\!\mid\exp (z^* c)\,,
\end{array}
\]
are different for~$H^+$ and~$H^-$. Here the short-hand notations $z^*\equiv(z^*_1,\dots,z^*_M)$ and $z\equiv(z_1,\dots,z_M)$ are used for the sets of independent Grassmann parameters $z_k,z^*_k$, $k\in {\mathcal M}$ (it is appropriate to omit the extra index~$\pm$ in $z^*$, $z$). Besides, $\sum_{k=1}^{M} c^{\dagger}_kz_k\equiv c^{\dagger} z$,
$\prod_{k=1}^M dz_k\equiv dz$. Let us calculate
$\mathcal G^{\pm}_{\Fa}Z^{\pm}_{\Fa}$ in~\eqref{12} using the representation of the trace in the Grassmann integration formalism~\cite{18}:
\begin{equation}
\mathcal G^{\pm}_{\Fa}Z^{\pm}_{\Fa}=\int dz\,dz^*\,e^{z^*z}
\langle z^*|e^Se^{-\lambda H^{\pm}}e^Te^{-\nu H^{\pm}}|z\rangle.
\label{13}
\end{equation}

In order to represent the right-hand side of this equality as the functional integral, let us introduce~$L$ new copies of the coherent states $|x(I)\rangle$, $\langle x^*(I)|$, where $I\in\{1,2,\dots,L\}$. Each of the~$2L$ multi-indices $x^*(I)$, $x(I)$ is expressed by $M$ independent Grassmann parameters. Using the decompositions of unity one can represent the right-hand side of~\eqref{13} as the $(L+1)$-fold multiple integral. In order to express the quasi-periodicity condition it is convenient to introduce the auxiliary variables:
\begin{equation}
-{\widehat E}x(0)=x(L+1)\equiv z,\qquad
-x^*(L+1)=x^*(0){\widehat E}^{-1}\equiv z^*,
\label{14}
\end{equation}
where ${\widehat E}\equiv e^{\widehat S}e^{-\lambda{\widehat H}^{\pm}}e^{\widehat T}$. Tending~$L$ to infinity, we obtain the functional integral over the space of the trajectories $x^*(\tau)$, $x(\tau)$, where $\tau\in {\mathbb R}$:
\begin{equation}
\mathcal G^{\pm}_{\Fa}Z^{\pm}_{\Fa}=\int e^S\,d\lambda^*\,d\lambda
\prod_{\tau} dx^*(\tau)\,dx(\tau).
\label{15}
\end{equation}
The action functional $S\equiv\int L(\tau)\,d\tau$ is
expressed through the Lagrangian $L(\tau)$:
\begin{equation*}
L(\tau)\equiv x^*(\tau)\biggl(\frac{d}{d\tau}-
\widehat H^{\pm}\biggr)x(\tau)+
J^*(\tau)x(\tau)+x^*(\tau)J(\tau),
\end{equation*}
where
\begin{equation*}
J^*(\tau)\equiv\lambda^*(\delta(\tau)\hat I+
\delta(\tau-\nu){\widehat E}^{-1}),\qquad
J(\tau)\equiv(\delta(\tau)\hat I+
\delta(\tau-\nu)\widehat E)\lambda.
\end{equation*}
The integration over the auxiliary Grassmann variables $\lambda^*$,~$\lambda$ in~\eqref{15} guarantees the fulfilment of the
constraints~\eqref{14}. The $\delta$-functions in $J^*(\tau)$, $J(\tau)$ reduce $\tau\in\mathbb R$ to $\tau\in[0,\beta]$. The stationarity conditions $\delta S/\delta x^*=0$, $\delta S/\delta x=0$ result in the following regularized answer~\cite{18}:
\begin{equation}
\mathcal G^{\pm}_{\Fa}={\det}\biggl(\hat I+
\frac{e^{(\beta-\nu){\widehat H}^{\pm}}e^{\widehat S}
e^{-\lambda{\widehat H}^{\pm}}e^{\widehat T}-{\hat I}}
{\hat I+e^{\beta{\widehat H}^{\pm}}}\biggr).
\label{16}
\end{equation}
Furthermore, we substitute~\eqref{16} into~\eqref{10} and pass to the momentum representation. The procedure described can also be applied to other correlators $G_{j;l}^{ab}(t)$~\eqref{3}, where $a,b\in\{+,-\}$.

\section{Random walks}\label{sec3}

As it has been shown in~\cite{15},~\cite{16}, the flips of spins on a
one-dimensional lattice may be associated with a random movements
of walkers. Indeed, let us consider a state of the $XX$ Heisenberg
chain, which corresponds to the ferromagnetic ordering of $M$ spins: \break
$|\!\!\Uparrow\rangle\equiv\bigotimes_{n=1}^M|\!\!\uparrow\rangle_n$ (i.e., all spins are oriented ``up''). Consider the average of the following type:
\begin{equation}
F_{j;l}(\lambda)\equiv
\langle\Uparrow\!|\sigma_{j}^{+}e^{-\lambda H_0}
\sigma_{l}^{-}|\!\Uparrow\rangle,
\label{17}
\end{equation}
where the notation $H_0$ implies that the zero magnetic field
$h=0$ is taken in the Hamiltonian~\eqref{6},~\eqref{7} (we shall omit  the same subscript for the corresponding matrices ${\widehat H}^{\pm}$~\eqref{7}), and $\lambda\in\mathbb C$ is an ``evolution''
parameter. ``Up'' (or ``down'') direction of spin corresponds to the empty (or filled) site. Differentiating ${F_{j;l}({\lambda})}$~\eqref{17} and expanding the commutator $\lbrack H_0,\sigma_j^{+}\rbrack$ we obtain the difference--differential equation:
\begin{equation}
\frac{d}{d\lambda}\,F_{j;l}({\lambda})=
\frac12\,(F_{j+1;l}({\lambda})+F_{j-1;l}({\lambda}))
\label{18}
\end{equation}
(and similar equation can be also obtained for the fixed index
$j$). Solution of the given equation is specified by the boundary
conditions imposed on the lattice argument, and by the initial
condition at $\lambda=0$.

The average $F_{j;l}(\lambda)$ can be considered as the generating
function of the trajectories with random turns that start at the
$l$-th site and end up at the $j$-th site. Indeed, let us
introduce the notation~$\mathcal D^K_{\lambda}$ for the operator
of differentiation of $K$-th order with respect to~$\lambda$ at
the point $\lambda=0$. The application of~$\mathcal D^K_{\lambda}$
to the average~\eqref{17} leads to the answer:
\begin{equation}
\mathcal D^K_{\lambda}\bigl[F_{j;l}({\lambda})\bigr]=
\langle\Uparrow|\sigma_{j}^{+}(-H_0)^K\sigma_{l}^{-}|\Uparrow\rangle=
\sum_{n_1,\dots,n_{K-1}}\Delta^{(+)}_{jn_{K-1}}\dots
\Delta^{(+)}_{n_2n_1}\Delta^{(+)}_{n_1l}.
\label{19}
\end{equation}
The right hand side of~\eqref{19} coincides with the entry at the
crossing of the $j$-th row and the $l$-th column of the matrix
given by the product of $K$ copies of the hopping matrix~\eqref{2}. Each matrix in this product corresponds to a transition between the two nearest sites of the lattice. After multiplication by~$2^{K}$ (this is due to the accepted normalization of the matrix~\eqref{2}), the right-hand side of~\eqref{19} gives the number of the trajectories that consist of $K$ steps and are connecting the $l$-th and $j$-th sites. Let us denote this number by $|P_K(l\rightarrow j)|$.

Let $|P_K (l_1, \dots, l_N\rightarrow j_1, \dots, j_N)|$ be a
number of trajectories consisting of $K$ links made by $N$ vicious
walkers in the random turns model. Here, the initial and final
positions of the walkers on the sites are  given respectively by
the sequences $l_1 > l_2 > \dots > l_N$ and $j_1 > j_2 >\dots > j_N$.
Let us consider the $N$-point correlation function ($N\leq M$):
\begin{equation}
F_{j_1,j_2,\dots,j_N;l_1,l_2,\dots,l_N}({\lambda})=
\langle\Uparrow\!\!|\sigma_{j_1}^{+}\sigma_{j_2}^{+}\dots
\sigma_{j_N}^{+}e^{-\lambda H_0}\sigma_{l_1}^{-}
\sigma_{l_2}^{-}\dots\sigma_{l_N}^{-}|\!\!\Uparrow\rangle.
\label{20}
\end{equation}
The present correlator is related to enumeration of the admissible
trajectories which are traced by $N$ vicious walkers. Indeed, the
application of the operator $\mathcal D^K_{\lambda/2}$
to~\eqref{20} results in the average of the type
$$
\langle\Uparrow\!\!|\sigma_{j_1}^{+} \sigma_{j_2}^{+} \dots
\sigma_{j_N}^{+}(-2H_0)^K \sigma_{l_1}^{-}\sigma_{l_2}^{-} \dots
\sigma_{l_N}^{-} |\!\!\Uparrow\rangle\,.
$$
This average provides the numbers $|P_K (l_1, \dots, l_N\rightarrow j_1, \dots, j_N)|$ that can be established with the help of the commutator
\begin{equation}
\lbrack H_0,\sigma_{l_1}^{-}\sigma_{l_2}^{-}\dots\sigma_{l_K}^{-}]=
\sum_{k=1}^K\sigma_{l_1}^{-}\dots\sigma_{l_{k-1}}^{-}
[H_0,\sigma_{l_k}^{-}]\sigma_{l_{k+1}}^{-}\dots\sigma_{l_K}^{-}
\label{21}
\end{equation}
(in this case, differentiation with respect to $\lambda/2$, instead of~$\lambda$, allows to take into account the normalization of the hopping matrix~\eqref{2}). The condition of non-intersection of trajectories of the walkers is expressed by the vanishing of the correlation function~\eqref{20} for any pair of coinciding indices $l_k$ or $j_p$.

Differentiating~\eqref{20} with respect to~$\lambda$ and
applying~\eqref{21}, we obtain the equation:
\begin{equation}
\frac{d}{d\lambda}\,F_{j_1,\dots,j_N;l_1,\dots,l_N}(\lambda)=
\frac12\sum_{k=1}^N\bigl(F_{j_1,
\dots,j_N;l_1,l_2,\dots,l_k+1,\dots,l_N}({\lambda})+
F_{j_1,\dots,j_N;l_1,l_2,\dots,l_k-1,\dots,l_N}({\lambda})\bigr).
\label{22}
\end{equation}
Equation \eqref{22} has been considered in~\cite{16} for the case of periodicity with respect to the lattice argument and with the initial condition:
$$
F_{j_1, \dots, j_N;l_1, \dots, l_N}(0)\,=\,\prod_{m=1}^N\delta_{j_m,l_m}\,.
$$
The function $F_{j_1,j_2, \dots, j_N; l_1,l_2, \dots, l_N}(\lambda)$
can be expressed as the determinant of the matrix consisting of the averages of the type of~\eqref{17}~\cite{16}:
\begin{equation}
F_{j_1,\dots,j_N;l_1,\dots,l_N}({\lambda})=
\det\bigl(F_{j_r;l_s}({\lambda})\bigr)_{1\leq r,s\leq N}.
\label{23}
\end{equation}

\subsection{Random walks on the axis}
\label{sec3.1}
Let us consider an infinite chain ($M\to\infty$). Then, the modified Bessel function $I_{j-l}(\lambda)$ turns out to be a solution of equation~\eqref{18}, which respects the condition
$F_{j;l}({0})={\delta}_{j,l}$~\cite{15}:
\begin{equation}
F_{j;l}({\lambda})=I_{j-l}({\lambda})=\frac1{2\pi
}\int_{-\pi}^{\pi}d\theta\,e^{{\lambda}\cos\theta}e^{i(j-l)\theta}.
\label{24}
\end{equation}
There exists the following expansion into the power series for $I_{j-l}(\lambda)$:
\begin{equation}
I_{j-l}({\lambda})=\sum_{Q\geq|l-j|}
\frac{1}{\bigl(\frac{Q-j+l}2\bigr)!\,\bigl(\frac{Q+j-l}2\bigr)!}
\biggl(\frac{{\lambda}}2\biggr)^Q,
\label{25}
\end{equation}
where the summation index $Q$ is subjected to the requirement:
$Q+|j-l|=0\pmod{2}$. In the limit of large ``time''
(${\lambda}\rightarrow\infty$) and for moderate values of
$m\equiv|l-j|$, using the known asymptotics for the Bessel
function, we obtain for the generating function:
\begin{equation*}
F_{j;l}({\lambda})\simeq\frac{e^{{\lambda}}}{{\sqrt{2\pi{\lambda}}}}
\biggl(1-\frac{4m^2-1}{8{\lambda}}+\dotsb\biggr),
\end{equation*}
where the decay is governed by the critical exponent $\xi=-1/2$.

Let the number $K$ satisfies the relations $K\geq|l-j|$ and
${K+|j-l|=0\nobreak\pmod{2}}$. Then, differentiation of the
series~\eqref{25} leads to the binomial relation
${|P_K(l\rightarrow j)|}=C_K^L$ for the number of all lattice
paths of the ``length'' $K$ between two sites on the infinite axis:
\begin{equation}
|P_K(l\rightarrow j)|\equiv
\mathcal D^K_{\lambda/2}[F_{j;l}({\lambda})]=
\frac{(m+2L)!}{L!\,(m+L)!}\,.
\label{26}
\end{equation}
Here $L$ denotes the one-half of the total number of turns:
$L\equiv(K-m)/2$.

Let us consider now the multi-point correlation function $F_{j_1,j_2, \dots, j_N;l_1,l_2, \dots, l_N}({\lambda})$. As it has been shown above, $\mathcal D^K_{\lambda/2} [F_{j_1, \dots, j_N;l_1, \dots, l_N}(\lambda)]$ has the sense of the number of trajectories of
$N$ vicious walkers each of which has made $K$ steps. A different
combinatorial interpretation of this object, however,  can be proposed. Really, let us consider a representation of the multi-point correlator in the form of the determinant~\eqref{23}. Its entries $F_{j_r;l_s}(\lambda)$ in the case of an infinite chain are given by the Bessel function $I_{j_r-l_s}(\lambda)$~\eqref{24}. The operator $\mathcal D^K_{\lambda/2}$ acts on the determinant as the differentiation of the product of $N$ functions:
\begin{equation}
(f_1(x)f_2(x)\dots f_N(x))^{(K)}=
\sum_{n_1+n_2+\dots+n_N=K}P(n_1,n_2,\dots,n_N)
f^{(n_1)}_1f^{(n_2)}_2\dots f^{(n_N)}_N.
\label{27}
\end{equation}
The notation $f^{(n)}\equiv d^nf(x)/dx^n$ is used here, and the
coefficients $P(n_1,n_2,\dots,n_N)$ are the numbers of permutations with repeats:
\begin{equation}
P(n_1,n_2,\dots,n_N)\equiv
\frac{(n_1+n_2+\dots+n_N)!}{n_1!\,n_2!\,\dots n_N!}.
\label{28}
\end{equation}
Summation in~\eqref{27} is over all non-negative values of
$n_1,n_2,\dots,n_N$, provided their sum is equal to $K$.

Suppose further, that an $N$-dimensional (hyper-)cubic lattice of
infinite extension is given, and each site of this lattice is
labelled by a set of $N$ numbers. Let $\mathcal T_K(q_1,q_2, \dots
,q_N)$ be the number of the lattice trajectories that can be
traced by some walker from the ``initial'' point
$\textbf{\textit{O}}\equiv(0,0, \dots, 0)$ to a point
$(q_1,q_2,\dots,q_N)$ in $K$ steps (by a single step the walker
can move to one of the nearest sites). Let all numbers $q_k$ be
non-negative, and let the inequality $K\geq q_1+q_2+\dots+q_N$ be
fulfilled, which means that the steps that can compensate each
other are allowed. Let us denote the number of these steps as~$2L$,
\begin{equation}
L\equiv\frac{K-q_1-q_2-\dots-q_N}{2}\,.
\label{29}
\end{equation}
Taking into account~\eqref{29}, the following formula for
the number of paths takes place:
\begin{equation}
\mathcal T_K(q_1,q_2,\dots,q_N)=\sum_{L_1+L_2+\dots+L_N=L}
P(q_1+L_1,q_2+L_2,\dots,q_N+L_N,L_1,L_2,\dots,L_N),
\label{30}
\end{equation}
where summation is taken over all non-negative values of $L_1,L_2,\dots,L_N$, provided that their sum is equal to $L$, and the formula~\eqref{28} for the number of permutations with repeats is used.

Turning back to the function  $F_{j_1,j_2,\dots,j_N;l_1,l_2,\dots,l_N}({\lambda})$
let us define the matrix $(n_{rs})_{1\leq r,s\leq N}$ with the entries $n_{rs}\equiv j_r-l_s$. Then, we arrive to the following
\begin{assertion}
The number of trajectories consisting of $K$ links, which are traced by $N$ vicious walkers on an axis, is expressed through the number of trajectories of the same ``length'' $K$, which are traced by a single walker travelling over sites of $N$-dimensional lattice of infinite extension:
\begin{equation}
\begin{aligned}
|P_K(l_1,\dots,l_N\rightarrow j_1,\dots,j_N)|&\equiv
\mathcal D^K_{\lambda/2}\bigl[F_{j_1,\dots,j_N;l_1,\dots,l_N}(\lambda)\bigr]=
\\
&=
\sum_{S_{a_1,a_2,\dots,a_N}}(-1)^{\mathcal P_S}
\mathcal T_K(n_{a_11},n_{a_22},\dots,n_{a_NN}),
\end{aligned}
\label{31}
\end{equation}
where summation is taken over all permutations $S_{a_1,a_2,\dots,a_N}\equiv$
$S(\begin{smallmatrix}1,&2,&\dots,&N \\ a_1,&a_2,&\dots,&a_N\end{smallmatrix})$ of the numbers $1,2,\dots,N$, and~$\mathcal P_S$ implies a parity of a specific permutation.
\end{assertion}

\begin{proof}
In order to verify~\eqref{31} one should develop the determinant~\eqref{23} by a row or by a column and then apply the induction using the relations \eqref{26}--\eqref{30}.
\end{proof}

Let us calculate, for instance,~\eqref{30} at $N=2$:
\begin{equation}
\mathcal T_K(q_1,q_2)=C_{q_1+q_2+2L}^{q_1+L}
\sum_{k=0}^LC^{L-k}_{q_1+L}C_{q_2+L}^{k}=C_K^{q_1+L}C_K^L,
\label{32}
\end{equation}
where $L=(K-q_{1}-q_{2})/2$ denotes one-half of the total number
of turns. Then, using~\eqref{31} we obtain:
\begin{equation}
\mathcal D^K_{\lambda/2}\bigl[F_{j_1,j_2;l_1,l_2}(\lambda)\bigr]=
\mathcal T_K(n_{11},n_{22})-\mathcal T_K(n_{21},n_{12})=
\begin{vmatrix}
C_K^L&C_K^{L+n_{21}}
\\[1mm]
C_K^L&C_K^{L+n_{11}}
\end{vmatrix},
\label{33}
\end{equation}
where $L=(K-n_{11}-n_{22})/2$ and the equality $n_{11}+n_{22}=n_{12}+n_{21}$ is used.

Representation of the entries $F_{j_1,j_2,\dots,j_N;l_1,l_2,\dots,l_N}(\lambda)$~\eqref{23}
in the integral form~\eqref{24} allows to obtain the following expression~\cite{15}:
\begin{equation}
\begin{aligned}
F_{j_1,\dots,j_N;l_1,\dots,l_N}({\lambda})={}&
\frac{e^{{\lambda}N}}{N!}\prod_{i=1}^N
\biggl(\,\int_{-\pi}^{\pi}\frac{d\theta_i}{2\pi}\biggr)
e^{-{\lambda}\sum_{k=1}^N(1-\cos\theta_k)}\times{}
\\
&\times
{S}_{\boldsymbol{\pi}}(e^{i\theta_1},e^{i\theta_2},\dots,
e^{i\theta_N})\prod_{1\leq j<k\leq N}
|e^{i\theta_j}-e^{i\theta_k}|^2,
\end{aligned}
\label{34}
\end{equation}
where $S_{\boldsymbol{\pi}}(e^{i\theta_1},e^{i\theta_2},
\dots,e^{i\theta_N})$~ is the Schur function~\cite{24},
\begin{equation}
S_{\boldsymbol{\pi}}(x_1,x_2,\dots,x_N)\equiv
\frac{\det(x_j^{\pi_k+N-k})_{1\leq j,k\leq N}}
{\det(x_j^{N-k})_{1\leq j,k\leq N}}\,.
\label{35}
\end{equation}
The Schur function \eqref{35} depends on the partition
$\boldsymbol{\pi}=(\pi_1,\pi_2,\dots,\pi_N)$ defined by a sequence
of non-negative integers, which are ordered according to
non-strict decreasing:  $\pi_1\geq\pi_2\geq\dots\geq\pi_N\geq0$.
In virtue of translational invariance it is always possible to
choose  the numbers $l_1$~$>$ $l_2>\dots>l_N\geq-N$ for the initial position of the walkers and to define the elements of the partition  by the equalities $\pi_k=l_k+k$. In order to calculate the leading asymptotics of the generating function in the limit ${\lambda}\rightarrow\infty$, let us transform the integral~\eqref{34} into the following integral~\cite{7},~\cite{25}:
\begin{equation*}
\int d^n\theta\prod_{1\leq j<k\leq N}|\theta_j-\theta_k|^2
e^{-{\lambda}/2\sum_{k=1}^N\theta^2_k}=
\frac{(2\pi)^{N/2}}{{\lambda}^{N^2/2}}\biggl(\,\prod_{p=1}^Np!\biggr).
\end{equation*}
It is a special case of the  \textit{Mehta integral}, which arises in
the theory of the \textit{Gaussian matrix ensembles}. Finally, we
obtain the following asymptotics of the generating function for
the trajectories traced by $N$ vicious walkers:
\begin{equation*}
F_{j_1,\dots,j_N;l_1,\dots,l_N}({\lambda})\simeq
\mathcal A\frac{e^{{\lambda}N}}{{\lambda}^{N^2/2}}\,,\qquad
\mathcal A=\frac{\prod_{p=1}^{N-1}p!}{(2\pi)^{N/2}}
\prod_{1\leq j<k\leq N}\frac{l_j-l_k}{k-j}\,,
\end{equation*}
where the well known formula for ${S}_{\boldsymbol{\lambda}}(1,1,\dots,1)$ is taken into account in $\mathcal A$ \cite{9},~\cite{25}. Therefore, the power-like behavior of $F_{j_1,\dots,j_N;l_1,\dots,l_N}({\lambda})$ is characterized by the exponent $\xi=-N^2/2$.

\subsection{Random walks over superposition of the eigen-states}
\label{sec3.2}
The eigen-functions of the $XX$ Hamiltonian, given by the relations \eqref{6},~\eqref{7}, are constructed as combinations of the states, obtained by ``flipping'' of $N$ spins in the state $|\!\!\Uparrow\rangle$~\cite{21}. Indeed, let us consider all admissible strict partitions $\boldsymbol{\mu}=(\mu_1,\mu_2, \dots, \mu_N)$,
where $M\geq\mu_1>\mu_2> \dots >\mu_N\geq1$, and establish a correspondence between each partition and an appropriate sequence of zeros and unities:
$\bigl\{e_k\equiv e_k(\boldsymbol{\mu})\bigr\}_{k\in\mathcal M}$, where
$e_k=\delta_{k,\mu_n}$, $1\leq n\leq N$. The required eigen-function
is defined as:
\begin{equation}
|\Psi_N(u_1,\dots,u_N)\rangle=
\sum_{\{e_k(\boldsymbol{\mu})\}_{k\in\mathcal M}}
\Upsilon_N(\{u_k\}\!\mid\boldsymbol{\mu})
(\sigma_{M}^{-})^{e_{M}}(\sigma_{M-1}^{-})^{e_{M-1}}\dots
(\sigma_{1}^{-})^{e_1}|\Uparrow\rangle,
\label{36}
\end{equation}
where summation is taken over all strict partitions $\boldsymbol{\mu}$ of the given type. The number of such partitions is expressed through the number of permutations with repeats~\eqref{28}: $P(N,M-N)=C_{M}^N$. The wave functions satisfy the periodic boundary conditions,
\begin{equation}
\Upsilon_N(\{u_k\}|\boldsymbol{\mu})\equiv
\det(u_k^{2\mu_l})_{1\leq k,l\leq N}
\label{37}
\end{equation}
are parametrized by the partitions $\boldsymbol{\mu}$ and by different, up to permutation, sets $\{u_1, \dots, u_N\}$ of solutions of the Bethe equations:
\begin{equation}
u_k^{2M}=(-1)^{N-1},\qquad1\leq k\leq N.
\label{38}
\end{equation}
These solutions have the form: $u_k^2=e^{i2\pi I_k/M}$, where
$I_k$ are integers or half-integers (this depends on parity of
$N$). Due to the antisymmetry of~\eqref{36} with the respect to
permutations of the parameters $u_k$, it is sufficient to restrict
oneself to the strict partitions $M\geq I_1>I_2>\dots>I_N \geq1$
in order to guarantee the single-valuedness of $u_k$. With the
help of~\eqref{36} the corresponding normalized average
\begin{equation}
\langle\sigma_{m+1}^{+}e^{-{\lambda}H_0}\sigma_1^{-}\rangle_N\equiv
\frac{\langle\Psi_N|\sigma_{m+1}^{+}e^{-{\lambda}H_0}\sigma_1^{-}
|\Psi_N\rangle}{\langle\Psi_N\mid\Psi_N\rangle}
\label{39}
\end{equation}
can be represented as a linear combination of $(N+1)$-point generating functions~\eqref{20}. Therefore, this average is related to the number of random walks of  $N+1$ pedestrians. The initial and the final positions of one of them are fixed at $l_1=1$ and $j_1=m+1$, respectively, while for the rest (\textit{virtual}) pedestrians the choice of their initial and the final positions is arbitrary.

Calculation of equation~\eqref{39} is of interest in the thermodynamic limit, when $M$ and $N$ are growing (their ratio remains finite), which means that the number of virtual pedestrians is increasing. In this limit~\cite{26}
\begin{equation}
\widetilde F_{m+1;1}({\lambda})\equiv
\langle\sigma_{m+1}^{+}e^{-{\lambda}H_0}\sigma_1^{-}\rangle_N
\bigr|_{M,N\gg1}\stackrel{\defa}{=}\Tr^{\prime}(\sigma_{m+1}^{+}
e^{-{\lambda}H_0}\sigma_1^{-}),
\label{40}
\end{equation}
where the notation $\Tr^{\prime}$ points out that the procedure
presented in Section~\ref{sec2} is used for the calculation of the
normalized average. The difference-differential relation, analogous to the equation~\eqref{18} is valid for
${\widetilde F}_{m+1;1}({\lambda})$~\eqref{40}:
\begin{align}
&\frac d{d{\lambda}}\,{\widetilde F}_{m+1;1}({\lambda})=
\frac12\,({\widetilde F}_{m;1}({\lambda})+
{\widetilde F}_{m+2;1}({\lambda}))-
\Tr^{{\prime}}(H_0\sigma_{m+1}^{+}e^{-{\lambda}H_0}\sigma_{1}^{-})-{}\notag
\\
&\qquad{}-
\Tr^{\prime}\biggl(\biggl(\frac{1-\sigma_{m+1}^{z}}{2}\biggr)
\sigma_{m}^{+}e^{-{\lambda}H_0}\sigma_{1}^{-}\biggr)-
\Tr^{\prime}\biggl(\biggl(\frac{1-\sigma_{m+1}^{z}}{2}\biggr)
\sigma_{m+2}^{+}e^{-{\lambda}H_0}\sigma_{1}^{-}\biggr).
\label{41}
\end{align}
The form of the present equation makes it possible  to suppose that the average ${\widetilde F}_{m+1;1}({\lambda})$ can also be of interest as a generating function of the random walks.

Let us turn to the calculation of ${\widetilde F}_{m+1;1}({\lambda})$ \eqref{40} in the fermionic representation~\eqref{4}. It is convenient to reduce the problem to calculation of the generating function of the form:
\begin{equation}
\widetilde{\mathcal G}\equiv\Tr^{\prime}
(e^Sc_{m+1}e^{-\lambda H_{0}}c^{\dagger}_1e^{-\nu H_{0}}),
\label{42}
\end{equation}
where the operator $S$ is defined just like in~\eqref{9} (i.e., by
means of the matrix $\widehat S = \break \diag\{S_1,S_2,\dots,S_M\}$).
Indeed, the functional ${\widetilde F}_{m+1;1}({\lambda})$
corresponds to the choice of $\nu=0$ and $S_k=-i\pi\theta(m-k)$,
where $\theta(m-k)$ is the Heavyside function, $\theta(0)=1$. The
second term in the right-hand side of~\eqref{41} corresponds to
the differentiation by $\nu$ at the point $\nu=0$, in the third
term we put $\nu=0$ and differentiate with respect to $S_{m+1}$.
In both cases we put $S_k=-i\pi\theta(m-k)$. Taking into account
the fact, that the contribution of the terms labelled by the
index~$\Ba$ in~\eqref{11} is negligible at sufficiently large $M$
and $N$, we approximately obtain:
\begin{equation}
\begin{aligned}
&\widetilde{\mathcal G}\approx
\biggl[\tr(e^{-{\lambda}{\widehat H}^{0}}{\hat e}_{1,m+1})-
\frac d{d\alpha}\biggr]\det({\hat I}+\widehat{\mathcal M}_1+
\alpha\widehat{\mathcal M}_2)\bigr|_{{\alpha}=0},
\\
&\widehat{\mathcal M}_1+\alpha\widehat{\mathcal M}_2\equiv
e^{-\nu\widehat H^{0}}e^{\widehat S}e^{-\lambda{\widehat H}^{0}}
({\hat I}+\alpha\hat e_{1,m+1}e^{-\lambda{\widehat H}^{0}}),
\end{aligned}
\label{43}
\end{equation}
where $\hat{ e}_{1,m+1}\equiv
({\delta}_{1,n}{\delta}_{m+1,l})_{1\leq n,l\leq M}$. The matrix
$\widehat H^{0}$ is used instead of $\widehat H^{\pm}$ since $s$
can be replaced by zero for the sufficiently large $M$.

The relation~\eqref{43} is written in the coordinate representation. In order to pass to the momentum representation it is convenient to use certain formulas provided in~\cite{22}. Keeping the matrix notations as in~\eqref{43}, we obtain the answer for ${\widetilde F}_{m+1;1}({\lambda})$ (in the limit $M\to\infty$, the corresponding operations should be understood in the sense of the operations over the corresponding integral operators~\cite{26}):
\begin{equation}
{\widetilde F}_{m+1;1}({\lambda})=
\det({\hat I}+\widehat{\mathcal U}_m)
\biggl[\tr(e^{-{\lambda}\hat{\varepsilon}_{0}}\breve e_{1,m+1})-
\tr\biggl(\frac{\widehat{\mathcal V}_{m}}{{\hat I}+
\widehat{\mathcal U}_m}\biggr)\biggr]
\label{44}
\end{equation}
(the notation $\tr$, for instance, corresponds to the trace of
$M\times M$ matrices). The matrices $\widehat{\mathcal U}_m$,
$\widehat{\mathcal V}_{m}$, ${\breve e}_{1,m+1}$ are given by the
corresponding momentum representations of the matrices
$\widehat{\mathcal M}_1$, $\widehat{\mathcal M}_2$, ${\hat
e}_{1,m+1}$~\eqref{43}. However, we shall need explicit
expressions only for the traces $\tr\widehat{\mathcal U}_m$ and
$\tr\widehat{\mathcal V}_{m}$ (see below). In the momentum
representation, $\hat{\varepsilon}_{0}$ is a diagonal matrix of
the eigen-energies of the $XX$ model at $h=0$~\cite{21}. Formally
expanding ${\widetilde F}_{m+1;1}({\lambda})$ in the powers of
$\widehat{\mathcal U}_m$  we shall obtain the answer in two lowest
orders:
\begin{align}
&{\widetilde F}_{m+1;1}({\lambda})\approx
F_{m+1;1}({\lambda})+F_{m+1;1}({\lambda})
\tr\widehat{\mathcal U}_m-\tr\widehat{\mathcal V}_{m},
\nonumber
\\
&\tr\widehat{\mathcal U}_m=(M-2m)F_{1;1}({\lambda}),
\label{45}
\\
&\tr\widehat{\mathcal V}_{m}=F_{m+1;1}(2{\lambda})-
2\sum_{l=1}^{m}F_{m+1;l}({\lambda})F_{l;1}({\lambda}),
\nonumber
\end{align}
where the notation $F_{j;l}({\lambda})$ implies the relations~\eqref{24}. Although $M$ and $m$ are chosen to be large
enough, the ratio $m/M$ is assumed to be finite.
Equation~\eqref{41} is fulfilled in each order separately by the
terms presented in~${\widetilde F}_{m+1;1}$~\eqref{45}.

By an analogy with the  ferromagnetic case, let us act
on~${\widetilde F}_{m+1,1}({\lambda})$~\eqref{45} by the
operator~$\mathcal D^K_{\lambda/2}$. Then, in the first order we
shall obtain the relation~\eqref{26}. In the second order, the
answer is of the following form:
\begin{equation}
MC_K^LC_K^L-C_K^L\sum_{l=0}^KC_K^l+2\sum_{l=1}^m
\begin{vmatrix}
C_K^{L+l-1}&C_K^L
\\
C_K^L&C_K^L
\end{vmatrix}.
\label{46}
\end{equation}
By virtue of~\eqref{33}, the result of the application of
$\mathcal D^K_{\lambda/2}$ to the second order function
$F_{j_1,j_2;l_1,l_2}(\lambda)$ is connected, as a particular case
of~\eqref{31}, with the number of the two-dimensional paths
$\mathcal T_K$, and is expressed through the corresponding
determinant. It means that it will be appropriate to
express~\eqref{46} in the following equivalent form:
\begin{align}
&(M-K)|P_{K}(l\rightarrow l+m)|^2+{}
\notag \\
&\qquad{}+
\mathcal D^K_{\lambda/2}\begin{bmatrix}
2\displaystyle\sum_{l=1}^m
\begin{vmatrix}
F_{m+1;l}({\lambda})&F_{m+1;1}({\lambda})
\\
F_{l;l}({\lambda})&F_{l;1}({\lambda})
\end{vmatrix}
-\displaystyle\sum_{l=0}^K
\begin{vmatrix}
F_{m+L;l}({\lambda})&F_{m+1;1}({\lambda})
\\
F_{l;l}({\lambda})&F_{l;L}({\lambda})
\end{vmatrix}
\end{bmatrix}.
\label{47}
\end{align}
In other words, the result of application of $\mathcal D^K_{\lambda/2}$ to~\eqref{45} in the second order can be reformulated in terms of the random walks of the two pedestrians (see~\eqref{23} and~\eqref{33}) and the squared number of walks of a single pedestrian. The summation by  the index $l$ in $\tr\widehat{\mathcal V}_{m}$~\eqref{45} can be interpreted as the summation over positions of a virtual walker in~\eqref{47}.

Using the equation~\eqref{33}, one can represent~\eqref{47} in
terms of the number of trajectories on a two-dimensional lattice:
\begin{equation}
(M-2(m+1))\mathcal T_K(m,0)+
\sum_{l=0}^m\mathcal T_K(m-l,l)-\sum_{l=1}^L\mathcal T_K(m+l,l)-
\sum_{l=1}^L\mathcal T_K(l,m+l).
\label{48}
\end{equation}
In this relation various lattice trajectories of $K$ links are
enumerated. All these trajectories start at the same point
$\textbf{\textit{O}}=(0,0)$ while they terminate on the segments
of the dashed broken line which connects the points $(L,L+m)$,
$(0,m)$, $(m,0)$, and $(L+m,L)$ (see figure). Formally, the sign
of the sum is not definite though its asymptotics is positive, in
general. An analogous description is expected in the higher orders
as well.

\begin{figure}[h!]
\centering

\includegraphics{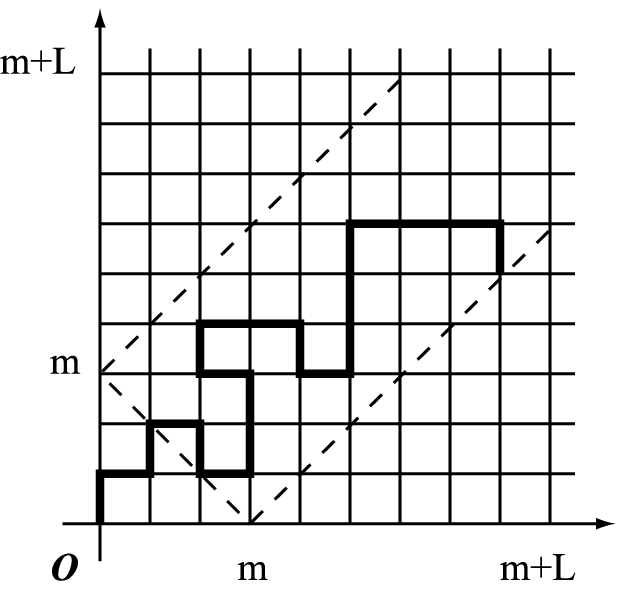}
{\small

\centerline{Typical configuration.}

}

\end{figure}

Let us estimate the behavior of the number of paths, which is
given by the representation~\eqref{46}, in the limit, when the
number of links $K=m+2 L$ increases. We shall assume, that the
restriction $1\ll m\ll L$ is valid which means that $m$ increases
moderately in the comparison with the increase of the number of
turns $L$: for instance, let $L$ increase as $m^2$. Using the
known asymptotical expansion of the logarithm of the
gamma-function (see Appendix)~\cite{27}, one can estimate the binomial coefficient $C_K^L$. Restricting ourselves by the first order of smallness, we obtain:
\begin{equation}
C_K^L\approx\frac{2^K}{{\sqrt{\pi L}}}e^{-m^2/(4L)}
\biggl(1-\frac{m}{2L}\biggl(1-\frac{m^2}{4L}\biggr)\biggr)\approx
\frac{2^K}{{\sqrt{\pi L}}}\biggl(1-\frac{m^2}{4L}\biggr)
\sim\frac{2^{2 L}}{{\sqrt{\pi L}\,}}.
\label{49}
\end{equation}
The second approximate equality in~\eqref{49} takes place if $L$
is increasing faster than $m^2$. The estimate~\eqref{49} characterizes an increase of the number of the trajectories~\eqref{26} for a single pedestrian.

The third term in~\eqref{46} can be written as $2A(m,L)C_K^LC_K^L$, where
\begin{equation}
A(m,L)\equiv-m+\sum_{l=1}^m\frac{(L+m+2-l)_{l-1}}{(L+1)_{l-1}}\,.
\label{50}
\end{equation}
Standard notation $(\alpha)_{n}$ for the Pochhammer's symbol is used
in \eqref{50}~\cite{27}. Applying again an expansion of the
logarithm of the gamma-function (A1), we can estimate $A(m,L)$~\eqref{50}:
\begin{equation}
A(m,L)\simeq mZ_1(m,L)-Z_0(m,L)+\mathcal{O}(m^{-1}),
\label{51}
\end{equation}
where
\begin{equation}
\begin{aligned}
&Z_0(m,L)\equiv e^{m^2/4L}\biggl(1+\frac{m}{L^2}\sum_{l=0}^{m/2}
e^{-l^2/L}\biggl(\frac{m^2}{4}-l^2\biggr)\biggr),
\\
&Z_1(m,L)\equiv-1+e^{m^2/4L}\,\frac{2}{m}\sum_{l=0}^{m/2}e^{-l^2/L}.
\end{aligned}
\label{52}
\end{equation}
Let  the values  $m$ and $L$ increase with the ratio $L/m^2$ being
finite and of order of unity. It can be shown (by means of numerical check as well), that the coefficient functions $Z_0(m,L)$ and $Z_1(m,L)$ remain finite in this case, and the contribution of $Z_0(m,L)$ is negligible  in comparison with $m Z_1(m,L)$ in~\eqref{51}. One can use Eqs.~\eqref{49} and~\eqref{51} in order to estimate~\eqref{46} in the leading approximation:
\begin{equation}
\frac{2^{4L}}{\pi L}\,e^{-m^2/2L}(M+2mZ_1(m,L)-
(\pi L)^{1/2}e^{m^2/4L}).
\label{53}
\end{equation}
Because of the behavior of the coefficient $Z_1(m,L)$, the
corresponding contribution in~\eqref{53} may turn out to be
comparable with $M$. The relation~\eqref{53} demonstrates that the
description of the random walks considered in the representation
of the superposition of the eigen-states is more complicated than
the one in the ferromagnetic case. This description can be
regarded as a simultaneous walks of the initial (i.e., principal)
and virtual pedestrians. The ending points of the trajectories
belonging to all the three segments of the dashed broken line on
the figure (see the representation of two-dimensional random
walks~\eqref{48}) correspond to comparable contributions into the
estimate~\eqref{53}. In certain cases, characterized by the
limiting behavior of the ratio $m^2/L$, the contribution of the
segment between the points $(m,0)$ and $(0,m)$ can become dominating.

\section{Conclusion}
\label{sec4}

It is shown that the correlation functions of the $XX$ Heisenberg
magnet, calculated over the super\-po\-si\-tion of the
eigen-states, as well as over the ferromagnetic state, are
connected with enumeration of the trajectories made by the walkers
moving on the lattice. A relationship is established between the
number of trajectories made by a several vicious walkers and the
number of paths made by a single random turns walker on a lattice
of a dimension equal to the number of the vicious walkers.
Differentiation of the generating function, calculated over the
superposition of the eigen-states, demonstrates a more complicated
combinatorial picture than that of the ferromagnetic case. In
particular, the set of the paths made by a single pedestrian is
replaced by the set of trajectories made simultaneously by the
principal and virtual (both vicious) pedestrians. An estimate is
obtained for the number of trajectories made both by the
principal and the virtual pedestrians.

\section*{Acknowledgement}

This paper was partially  supported  by the Russian Foundation for
Basic Research, No.~07-01-00358, and by the Russian Academy of
Sciences program ,,Mathematical Methods in Non-Linear Dynamics''.

\section*{Appendix}

Asymptotic expansion for the logarithm of the gamma-function at large $|z|$ and $|\arg z |< \pi$ has the form \cite{27}:
$$
\begin{array}{rcl}
\log\Gamma
(z\,+\,\alpha)&=&\displaystyle{\Bigl(z\,+\,\alpha\,-\,\frac12\Bigr) \log
z\,-\,z\,+\,\frac12\,\log(2\pi)}\nonumber\\[0.4cm]
&+&\displaystyle{\sum\limits_{p=1}^n
(-1)^{p+1}\,\frac{B_{p+1}(\alpha)}{p(p+1)}\,z^{-p}\,+\,
\mathcal{O}\Bigl(\frac1{z^{n+1}}\Bigr)}\,,
\end{array}\eqno(A1)
$$
where $n= 1, 2, 3, \dots$. The Bernoulli polynomials $B_{n}(\alpha)$ ($A1$) are defined as follows:
\[
B_n(\alpha)\,=\,\sum\limits_{l=0}^n C_n^l\,B_l\,\alpha^{n-l}\,,
\]
where $C_n^l$ are the binomial coefficients, and $B_l$ are the Bernoulli numbers. The first Bernoulli polynomials $B_{p}(\alpha)$ look as follows:
\begin{eqnarray}
B_0(\alpha)\,=\,1\,,\qquad B_1(\alpha)\,=\,\alpha\,-\,\frac12\,,\qquad
B_2(\alpha)\,=\,\alpha^2\,-\,\alpha\,+\,\frac16\,,\nonumber\\[0.4cm]
B_3(\alpha)\,=\,\alpha^3\,-\,\frac32\,\alpha^2\,+\,\frac12\,\alpha\,,\qquad
B_4(\alpha)\,=\,\alpha^4\,-\,2 \alpha^3\,+\,\alpha^2\,-\,\frac1{30}\,,
\nonumber\end{eqnarray}
where the Bernoulli numbers $B_l$ are used:
\begin{eqnarray}
B_0\,=\,1\,,\qquad B_1\,=\,-\,\frac12\,,\qquad
B_2\,=\,\frac16\,,\qquad B_4\,=\,-\,\frac1{30}\,.
\nonumber
\end{eqnarray}

\end{document}